\documentclass{article}

\usepackage{PRIMEarxiv}

\usepackage[utf8]{inputenc} 
\usepackage[T1]{fontenc}    
\usepackage{hyperref}       
\usepackage{url}            
\usepackage{booktabs}       
\usepackage{amsfonts}       
\usepackage{nicefrac}       
\usepackage{microtype}      
\usepackage{lipsum}
\usepackage{fancyhdr}       
\usepackage{graphicx}       
\graphicspath{{media/}}     

\usepackage{authblk}
\usepackage{cite}
\usepackage{amsthm}
\usepackage{amsmath,amssymb,amsfonts}
\usepackage{algorithm}
\usepackage{algorithmic}
\usepackage{graphicx}
\usepackage{textcomp}
\usepackage{enumitem}
\usepackage{url}
\usepackage{eqnarray}
\usepackage{bm}
\usepackage{array}
\usepackage{color}
\newtheoremstyle{mystyle}
{}
{}
{\itshape}
{}
{\bfseries}
{.}
{ }
{\thmname{#1}\thmnumber{ #2}\thmnote{ (#3)}}%

\theoremstyle{mystyle}
\newtheorem{observation}{Observation}
\newtheorem{definition}{Definition}
\newtheorem{theorem}{Theorem}

\title{Burning graphs through farthest-first traversal
}

\author[a,b]{\small Jesús García Díaz}
\author[a,b]{\small Julio César Pérez Sansalvador}
\author[a,b]{\small Lil María Xibai Rodríguez Henríquez}
\author[b]{\small José Alejandro Cornejo Acosta}

\affil[a]{\footnotesize Consejo Nacional de Ciencia y Tecnología, Mexico city 03940, Mexico}
\affil[b]{\footnotesize Instituto Nacional de Astrofísica, Óptica y Electrónica, Coordinación de Ciencias Computacionales, Puebla 72840, Mexico}

\begin{document}
\maketitle

\begin{abstract}
The graph burning problem is an NP-hard combinatorial optimization problem that helps quantify the vulnerability of a graph to contagion. This paper introduces a simple farthest-first traversal-based approximation algorithm for this problem over general graphs. We refer to this proposal as the Burning Farthest-First (BFF) algorithm. BFF runs in $O(n^3)$ steps and has an approximation factor of $3-2/b(G)$, where $b(G)$ is the size of an optimal solution. Despite its simplicity, BFF tends to generate near-optimal solutions when tested over some benchmark datasets; in fact, it returns similar solutions to those returned by much more elaborated heuristics from the literature.
\end{abstract}

\keywords{Approximation algorithms \and farthest-first traversal \and graph burning \and information spreading \and social contagion.}

\section{Introduction}
\label{sec:introduction}
The graph burning problem is an NP-hard combinatorial optimization problem introduced in 2014 as a contagion model in social networks \cite{bonato2014burning}. However, it also models other phenomena, such as the sequential spread of information on general networks and the spread of viral infections under a very idealistic context \cite{vsimon2019heuristics,gupta2021burning,gautam2021faster}. The main attribute of this problem is that it helps quantify how vulnerable a graph is to contagion. This problem's input is a simple graph $G=(V,E)$, and its goal is to find a minimum length sequence of vertices $(y_1,y_2,...,y_k)$ such that, by repeating the following steps from $i=1$ to $k$, all vertices in $V$ get burned \cite{bonato2014burning,bessy2017burning,bonato2019approximation}. In the beginning, all vertices are unburned.

\begin{enumerate}[label=\alph*)]
	\item The neighbors of the previously burned vertices get burned.
	\item Vertex $y_i$ gets burned.
\end{enumerate}

Every time a vertex gets burned, it remains in that state. A \textit{burning sequence} is a sequence of vertices that burns all the input graph vertices by repeating the previous steps. An \textit{optimal burning sequence} is a burning sequence of minimum length. Thus, the graph burning problem seeks an optimal burning sequence. Figure~\ref{fig:0} shows the optimal burning sequence of a small graph. This figure emphasizes an important observation exploited by the graph burning problem: \textit{fire} can be spread from different \textit{places} at different \textit{times}. Also, the fire that started earlier spreads further. Specifically, at step 1a, the set of burned vertices is the empty set; therefore, no vertex gets burned. At step 1b, vertex $v_7$ gets burned. Afterwards, at step 2a, the neighbors $\{v_2,v_3,v_{10},v_{11}\}$ of the previously burned vertices $\{v_7\}$ get burned too. At step 2b, vertex $v_8$ gets burned. Finally, at step 3a, the neighbors $\{v_1,v_4,v_5,v_6,v_9,v_{12},v_{13},v_{14}\}$ of the previously burned vertices $\{v_2,v_3,v_7,v_{10},v_{11},v_8\}$ get burned too, and at step 3b, vertex $v_{15}$ gets burned. At this point, all vertices are burned. Thus, the sequence $(v_7,v_8,v_{15})$ is a burning sequence. Intuitively, a network with a smaller burning sequence is more vulnerable to get quickly burned. Identifying the elements of an optimal burning sequence may be helpful to encourage or avoid their burning, depending on if spreading the \textit{fire} is desirable or not. For example, for social contagion in social networks, it may be desirable to quickly spread particular messages to users. However, it may also be desirable to prevent fake news from spreading \cite{bonato2014burning}.

Due to its nature, the graph burning problem has been approached chiefly through approximation algorithms and heuristics. These heuristics are mainly based on different centrality measures; through experimentation, they have shown to be relatively efficient \cite{vsimon2019heuristics,gautam2021faster}. Regarding approximation algorithms, there are two 3-approximation algorithms for general graphs \cite{bessy2017burning,bonato2019approximation}, a 2-approximation algorithm for trees \cite{bonato2019approximation}, a 1.5-approximation algorithm for graphs with disjoint paths \cite{bonato2019approximation}, and a 2-approximation algorithm for square grids \cite{gupta2021burning}. This paper introduces two $3-2/b(G)$-approximation algorithms (Algorithms~\ref{alg:1} and \ref{alg:2}) for the graph burning problem over general graphs, where $b(G) \in \mathbb{Z}^+$ is the size of an optimal solution. The purpose of Algorithm~\ref{alg:1} is to smooth the way for understanding Algorithm~\ref{alg:2}. While Algorithm~\ref{alg:1} requires the optimal burning sequence length in advance, making it impractical, Algorithm~\ref{alg:2} does not have this requirement. Therefore, the latter is the main contribution of this paper. We refer to this algorithm as Burning Farthest-First (BFF), and its main attributes include conceptual simplicity and competitive performance.

\begin{figure}[]
	\centering
	\includegraphics[scale=0.47]{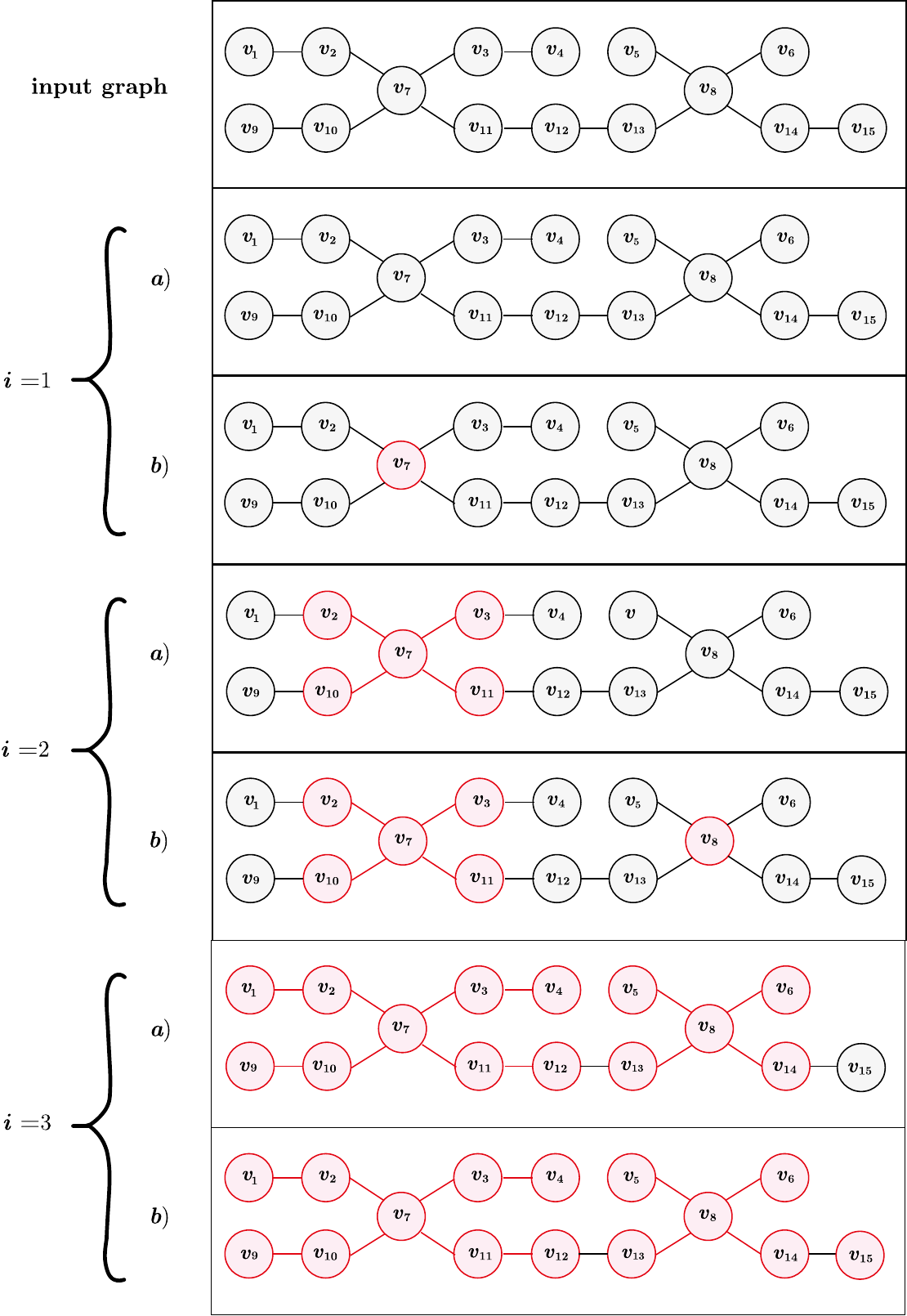}
	\caption{The optimal burning sequence of the given input graph is $(v_7,v_8,v_{15})$}
	\label{fig:0}
\end{figure}

Regarding the 3-approximation algorithms from the literature \cite{bessy2017burning,bonato2019approximation}, it is worth mentioning the following. Although a more detailed analysis may show that they return solutions of length at most $3 \cdot b(G)-\beta$, their authors did not report the exact value of $\beta$. Besides, even though the paper introducing Bonato and Kamali's approximation algorithm \cite{bonato2019approximation} mentions the values $3g-2$ and $3g-3$, none of them is necessarily the length of the returned solution. This is because their algorithm does not necessarily returns a burning sequence when $g=b(G)$ (See the proof of Lemma 2 from \cite{bonato2019approximation}: \textit{``Burn-Guess($G$,$g$) returns Bad-Guess if the number of centers becomes equal to $g$.''} So, if the input is a disconnected graph $G$ with $g$ components and $b(G)=g$, Burn-Guess($G$,$g$) will return Bad-Guess because it selects at least one vertex from each component.) In the analysis performed in this paper, we computed the exact value of $\beta$ for the BFF algorithm. Therefore, the analysis presented in this paper improves the analysis performed in previous works.

The remaining part of the document is organized as follows. Section \ref{sec:1} gives some background definitions, including the formal definition of the problem. Section \ref{sec:2} introduces the proposed approximation algorithms. This same section presents a running time analysis, a collection of tight examples, and a brief empirical evaluation. Finally, Section \ref{sec:con} presents the concluding remarks.

\section{Graph burning problem}
\label{sec:1}

Some context definitions are provided before presenting the formal definition of the graph burning problem. All of these refer to a simple graph $G=(V,E)$.

\begin{definition}
	The distance $d(u,v)$ between two vertices $u,v \in V$ is defined as the number of edges in the shortest path between them.
	\label{def:0}
\end{definition}

\begin{definition}
	The distance $d(u,S)$ between a vertex $u \in V$ and a set $S \subseteq V$ is defined as the distance from vertex $u$ to its nearest vertex in $S$. Namely, $d(u,S)=\min_{v \in S} d(u,v)$.
	\label{def:1}
\end{definition}

\begin{definition}
	The open neighborhood $N(v)$ of a vertex $v \in V$ is the set of vertices at distance one from $v$. Notice that $v \not \in N(v)$.
	\label{def:2}
\end{definition}

\begin{definition}
	The closed neighborhood $N[v]$ of a vertex $v \in V$ is the set of vertices at distance at most one from $v$. In other words, $N[v] = N(v) \cup \{v\}$.
	\label{def:3}
\end{definition}

\begin{definition}
	The closed $k^{th}$ neighborhood $N_{k}[v]$ of a vertex $v \in V$ is the set of vertices at distance at most $k$ from $v$. Notice that $N_0[v] = \{v\}$ and $N_1[v] = N[v]$.
	\label{def:4}
\end{definition}

\begin{definition}
	A finite sequence is a function $f:\{x \in \mathbb{Z}^+ \ | \ x \le k\} \rightarrow S$, where $k \in \mathbb{Z}^+$ is the length of the sequence, $S$ is the set of objects in the sequence, and $|S|\le k$. Notice that $f(i)$ is the object in the $i^{th}$ position of the sequence. Besides, each position is assigned to only one object, but one same object may be at many positions, i.e., $f$ is surjective but not necessarily injective.
\end{definition}

From now on, the word \textit{burn} is mostly replaced by the word \textit{cover}. We say that a vertex $v$ in a given burning sequence is responsible for \textit{covering} all the vertices in its closed $j^{th}$ neighborhood $N_{j}[v]$, where $j$ depends on the sequence's length and the position of vertex $v$. This value $j$ is referred to as the \textit{covering radius} of vertex $v$. As an instance, vertex $v_7$ of Figure~\ref{fig:0} covers all vertices in $N_{2}[v_7]=\{v_1,v_2,v_3,v_4,v_7,v_9,v_{10},v_{11},v_{12}\}$, vertex $v_8$ covers all vertices in $N_{1}[v_8]=\{v_5,v_6,v_8,v_{13},v_{14}\}$, and vertex $v_{15}$ covers all vertices in $N_{0}[v_{15}]=\{v_{15}\}$. Definitions \ref{def:5} and \ref{def:6} below formally describe the graph burning problem.

\begin{definition}
	The graph burning problem seeks a minimum length sequence of vertices $f^*: \{x \in \mathbb{Z}^+ \ | \ x \le b(G)\} \rightarrow S^*$ covering the whole set $V$ (Eq.~\ref{eq:cov:1}), where $b(G)$ (known as the \textit{burning number}) is its length, $S^* \subseteq V$, $|S^*|\le b(G)$, $f^{*-1}(v)$ is the set of positions that vertex $v$ has in the sequence, and $b(G)-\min f^{*-1}(v)$ is the bigger covering radius associated with vertex $v$.
	\begin{equation}
		\bigcup_{v \in S^*} N_{b(G) - \min f^{*-1}(v)}[v] = V
		\label{eq:cov:1}
	\end{equation}
	\label{def:5}
\end{definition}

By Definition \ref{def:5}, a burning sequence may have repeated vertices. The usual definition of the graph burning problem states that each vertex in a burning sequence is burned the first time by itself, which implies that all vertices in a burning sequence must be different \cite{bessy2017burning,bonato2019approximation,vsimon2019heuristics}. However, a burning sequence of optimal length may have repeated vertices (See Observation~\ref{obs:1} and Figure~\ref{fig:1}.) This observation will be important in the next section.

\begin{observation}
	A repeated vertex $v$ in a burning sequence $f:\{x\in \mathbb{Z}^+ \ | \ x\le k\} \rightarrow S$ has $|f^{-1}(v)|$ different covering radii. From these radii, the bigger one corresponds to $\min f^{-1}(v)$. Therefore, the presence of $v$ in positions $f^{-1}(v) \setminus \{\min f^{-1}(v)\}$ do not contribute to burning vertices; instead, it increases the covering radius of the vertices $u$ such that $\min f^{-1}(u) < \max f^{-1}(v)$. Thereby, even though a burning sequence with no repeated vertices may be better than a burning sequence with repeated vertices, the latter may exist; such as the ones in Figure~\ref{fig:1}.
	\label{obs:1}
\end{observation}

\begin{figure}[h]
	\centering
	\includegraphics[scale=0.5]{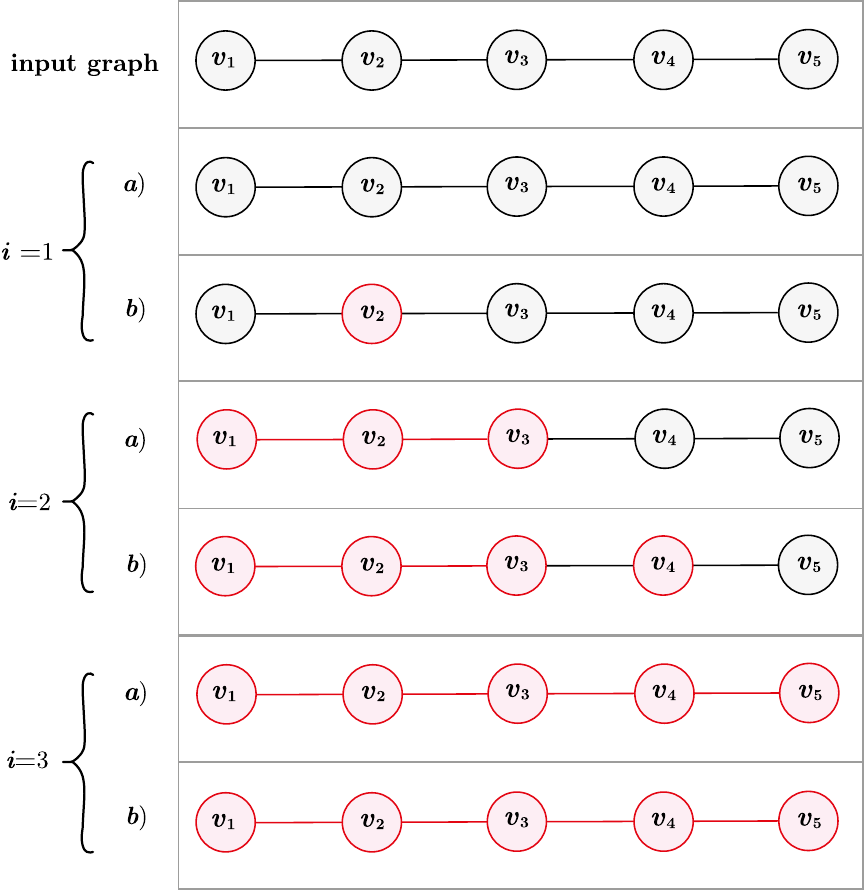}
	\caption{Vertices may get burned in the same fashion by different optimal burning sequences: $(v_2,v_4,v_1)$, $(v_2,v_4,v_2)$, $(v_2,v_4,v_3)$, $(v_2,v_4,v_4)$, and $(v_2,v_4,v_5)$; some of which have repeated vertices.}
	\label{fig:1}
\end{figure}

The following observation is used to rephrase the graph burning problem as a covering problem that receives as an input a complete weighted graph that follows a metric.

\begin{observation}
	
	The complete weighted graph \mbox{$G_w=(V,E_w)$} that results from computing the shortest path between every pair of vertices of a simple graph $G=(V,E)$ respects a metric. Namely,
	
	\begin{itemize}
		\item $\forall v \in V \ , \ d(v,v)=0$ ,
		\item $\forall u,v \in V \ , \ d(u,v)=d(v,u)$ ,
		\item $\forall u,v,w \in V \ , \ d(u,v) \le d(u,w) + d(w,v)$ ,
	\end{itemize}

	- where $d(u,v)$ is the weight of edge $\{u,v\} \in E_w$.
	
	Now, suppose $G$ is disconnected. In that case, the undefined distance between two vertices in different connected components can be substituted by $n$, which is greater than the maximum distance between any two connected vertices in $G$, where $n=|V|$. This way, the resulting weighted graph is completed and follows a metric too. Once the complete weighted graph $G_w=(V,E_w)$ is computed, the graph burning problem can be rephrased in the following way.
	
	\label{obs:2}
\end{observation}

\begin{definition}
	Given the complete weighted graph $G_w=(V,E_w)$ from Observation \ref{obs:2}, the graph burning problem seeks a minimum length sequence of vertices $f^*: \{x \in \mathbb{Z}^+ \ | \ x \le b(G)\} \rightarrow S^*$ covering the whole set $V$ (Eq.~\ref{eq:cov:2}), where $b(G)$ is its length, $S^* \subseteq V$, $|S^*|\le b(G)$, $f^{*-1}(v)$ is the set of positions that vertex $v$ has in the sequence, and $b(G)-\min f^{*-1}(v)$ is the bigger covering radius associated with vertex $v$.
	
	\begin{equation}
		\bigcup_{v \in S^*} \{ u \in V \ | \ d(u,v) \le b(G) - \min f^{*-1}(v) \}  = V
		\label{eq:cov:2}
	\end{equation}
	
	Observe that $b(G) \le n$, where $n=|V|$. Therefore, the covering radius of all vertices in $f^*$ cannot be greater than $n-1$. In consequence, the pairs of vertices at distance $n$ in $G_w$ (which are disconnected in the original input graph $G$) cannot burn each other. Thus, this definition is consistent with Definition \ref{def:5}.
	\label{def:6}
\end{definition}

Based on these definitions and observations, the following section introduces two approximation algorithms for the graph burning problem.

\section{Proposed approximation algorithms}
\label{sec:2}

This section introduces Algorithms~\ref{alg:1} and~\ref{alg:2}. Both algorithms run in $O(n^3)$ steps and have an approximation factor of $3-2/b(G)$, where $b(G)$ is the burning number of the input graph $G=(V,E)$. The goal of Algorithm~\ref{alg:1} is to smooth the way for understanding Algorithm~\ref{alg:2}, which we refer to as Burning Farthest-First (BFF). While Algorithm~\ref{alg:1} has the disadvantage of requiring the burning number $b(G)$ in advance, Algorithm~\ref{alg:2} does not. In both algorithms, the distance between every pair of vertices must be known in advance (See Observation~\ref{obs:2} and Definition~\ref{def:6}.) These distances may be computed by a polynomial-time shortest path algorithm, such as a Breadth-First Search-based algorithm (line 1) \cite{kleinberg2006algorithm}.

In summary, Algorithm~\ref{alg:1} begins by iteratively constructing a sequence $f:\{x\in \mathbb{Z}^+ \ | \ x\le 3\cdot b(G)-2\} \rightarrow S$ of exactly length $3\cdot b(G)-2$, where vertex $f(1)$ is selected at random and each other vertex $f(j)$, $1<j\le b(G)$, is the farthest one from the previously selected ones (lines 5 to 9). The rest of the vertices are selected at random (lines 10 to 14). The procedure for selecting the first $b(G)$ vertices has many applications and is known under different names, including farthest-first traversal, farthest point heuristic, and Gon algorithm \cite{rosenkrantz1977analysis,gonzalez1985clustering,dyer1985simple,xiang1997color,moenning2003new,mebarki2005farthest,huang2012farthest,garcia2017worse,garcia2019approximation,sheehy2020one}. The last $2\cdot b(G)-2$ random vertices have the effect of increasing the covering radius of the first $b(G)$ vertices. This way, each of the first $b(G)$ vertices in $f$ becomes responsible for covering all the vertices covered by some unique vertex in the optimal burning sequence $f^*:\{x\in \mathbb{Z}^+ \ | \ x\le b(G)\} \rightarrow S^*$. Therefore, $f$ covers all vertices in $V$, and it has a length of exactly $3\cdot b(G)-2$. Theorem~\ref{theo:1} shows the correctness of Algorithm~\ref{alg:1}.

Algorithm~\ref{alg:2} is similar to Algorithm~\ref{alg:1}. However, it has the advantage of not requiring the burning number $b(G)$ in advance. In summary, Algorithm~\ref{alg:2} constructs a sequence of vertices $f:\{x\in \mathbb{Z}^+ \ | \ x\le |S|\} \rightarrow S$ of length at most $3\cdot b(G)-2$, where vertex $f(1)$ is selected at random, and each other vertex $f(j)$, $1<j\le |S|$, is the farthest from the previously selected ones (lines 8 to 16). The difference between these algorithms is that Algorithm~\ref{alg:2} keeps adding the farthest vertex until all vertices in $V$ are covered (line 8). By Theorem~\ref{theo:1}, no more than $3\cdot b(G)-2$ vertices added this way are needed to cover all the vertices in $V$. Therefore, Algorithm~\ref{alg:2} returns a burning sequence of length at most $3\cdot b(G)-2$. Notice that this algorithm can return a burning sequence of length less than $3\cdot b(G)-2$, while Algorithm~\ref{alg:1} always returns a burning sequence of exactly length $3\cdot b(G)-2$. Theorem~\ref{theo:2} shows the correctness of Algorithm~\ref{alg:2}.

\begin{algorithm}[h]
	 \hspace*{\algorithmicindent} \textbf{Input:} A simple graph $G=(V,E)$ and $b(G)$ \\
	\hspace*{\algorithmicindent} \textbf{Output:} A burning sequence \\ 
	\hspace*{1.7cm}$f: \{x \in \mathbb{Z}^+ \ | \ x \le 3 \cdot b(G) - 2\} \rightarrow S$, \\ 
	\hspace*{1.7cm}\mbox{where $S\subseteq V$, and $|S|\le 3\cdot b(G)-2$}
\begin{algorithmic}[1]
	\STATE Compute all-pairs shortest path
	\STATE $v_1 \leftarrow $ any vertex in $V$
	\STATE $S \leftarrow \{v_1\}$
	\STATE $f \leftarrow (1,v_1)$
	\FOR{$i=2$ \textbf{to} $b(G)$}
		\STATE $v_i \leftarrow \arg \max_{u \in V} d(u,S)$
		\STATE $S \leftarrow S \cup \{v_i\}$
		\STATE $f \leftarrow f \cup (i,v_i)$
	\ENDFOR
	\FOR{$i=b(G)+1$ \textbf{to} $3\cdot b(G)-2$}
		\STATE $v_i \leftarrow \text{any vertex in } V$
		\STATE $S \leftarrow S \cup \{v_i\}$
		\STATE $f \leftarrow f \cup (i,v_i)$
	\ENDFOR
	\RETURN $f$ \;
\end{algorithmic}
	\caption{A $3-2/b(G)$-approximation algorithm for the graph burning problem with known burning number $b(G)$}
\label{alg:1}
\end{algorithm}

\begin{algorithm}[h]
	\hspace*{\algorithmicindent} \textbf{Input:} A simple graph $G=(V,E)$\\
	\hspace*{\algorithmicindent} \textbf{Output:} A burning sequence\\
	\hspace*{1.7cm}$f:\{x \in \mathbb{Z}^+ \ | \ x \le |S|\} \rightarrow S$,\\
	\hspace*{1.7cm}\mbox{where $S\subseteq V$, and $|S|\le 3\cdot b(G)-2$}
	\begin{algorithmic}[1]
		\STATE Compute all-pairs shortest path
		\STATE $v_1 \leftarrow $ any vertex in $V$
		\STATE $S \leftarrow \{v_1\}$
		\STATE $f \leftarrow (1,v_1)$
		\STATE $B_p \leftarrow \emptyset$
		\STATE $B_c \leftarrow \{v_1\}$
		\STATE $i \leftarrow 2$
		\WHILE{$B_c \not = V$}
			\STATE $v_i \leftarrow \arg \max_{u \in V} d(u,S)$
			\STATE $S \leftarrow S \cup \{v_i\}$
			\STATE $f \leftarrow f \cup (i,v_i)$
			\STATE $T \leftarrow N(B_c \setminus B_p)$
			\STATE $B_p \leftarrow B_c$
			\STATE $B_c \leftarrow B_c \cup T \cup \{v_i\}$
			\STATE $i \leftarrow i+1$
		\ENDWHILE
		\RETURN $f$
	\end{algorithmic}
	\caption{BFF}
	\label{alg:2}
\end{algorithm}

Before introducing Theorems~\ref{theo:1} and \ref{theo:2}, we introduce another observation. As before, this refer to a simple input graph $G=(V,E)$.

\begin{observation}
	Given an optimal burning sequence \mbox{$f^*:\{x \in \mathbb{Z}^+ \ | \ x \le b(G)\} \rightarrow S^*$}, the distance from every vertex $v \in V$ to the set $S^*$ is less than or equal to $b(G)-1$ (Eq.~\ref{eq:2} and \ref{eq:3}). Otherwise, $f^*$ would not be an optimal burning sequence.
	\label{obs:222}
\end{observation}

\begin{equation}
	\forall v \in V, \hspace{0.5cm} d(v,S^*) \le b(G)-1
	\label{eq:2}
\end{equation}

\begin{equation}
	\text{where} \hspace{0.4cm} d(v,S^*) = \min_{u \in S^*} d(v,u)
	\label{eq:3}
\end{equation}

\begin{theorem}
	Algorithm~\ref{alg:1} returns a burning sequence \mbox{$f:\{x \in \mathbb{Z}^+ \ | \ x \le 3\cdot b(G)-2\} \rightarrow S$}, where $S\subseteq V$ and $|S| \le 3\cdot b(G)-2$. Namely, Algorithm~\ref{alg:1} is a $3-2/b(G)$-approximation algorithm for the graph burning problem.
	\label{theo:1}
\end{theorem}

\begin{proof}
	First of all, observe the following facts. The length of the sequence $f$ returned by Algorithm~\ref{alg:1} is exactly $3\cdot b(G)-2$. Each of the first $b(G)$ vertices $f(i)$ is the farthest from the previously added vertices (lines 2 to 9). Namely, for every $i \in \{1,2,...,b(G)\}$, $d(f(i),\{f(1),f(2),...,f(i-1)\})$ is maximal. Notice that the first vertex is selected at random (line 2). The remaining $2 \cdot b(G) - 2$ vertices are added at random, too (lines 10 to 14).
	
	The general structure of the proof is the following. We show that each of the first $b(G)$ vertices added to the sequence covers all the vertices that some vertex $u^*_i \in S^*$ at a distance not greater than $b(G)-1$ covers, where $f^*:\{x \in \mathbb{Z}^+ \ | \ x \le b(G) \} \rightarrow S^*$ is an optimal burning sequence. Then, we show that each of these vertices $u^*_i$ is unique for each vertex $f(i)$. This way, the first $b(G)$ vertices in $f$ cover all vertices.
	
	First, by Observation~\ref{obs:222}, there is a vertex $u^*_i \in S^*$ at a distance not greater than $b(G)-1$ from any vertex $f(i)$. 
	
	\begin{equation}
		\exists u^*_i \in S^* \ , \  d(f(i),u^*_i) \le b(G)-1
		\label{eq:4}
	\end{equation}
	
	Now, the covering radius of any vertex $u^*_i$ is equal to $b(G)-\min f^{*-1}(u^*_i)$. However, its existence is all we know about $u^*_i$; thus, its position $\min f^{*-1}(u^*_i)$ is unknown. For this reason, let us assume the worst scenario, where $\min f^{*-1}(u^*_i)=1$. Under this assumption, vertex $u^*_i$ covers all the vertices in $N_{b(G)-1}[u^*_i]$. Since the edges of the complete graph of distances respect the triangle inequality (Observation~\ref{obs:2}), the distance from every vertex $v \in N_{b(G)-1}[u^*_i]$ to vertex $f(i)$ is less than or equal to $2\cdot (b(G)-1)$.

	\begin{equation}
	\begin{split}
		\forall v\in N_{b(G)-1}[u^*_i] \ , \ d(v,f(i)) & \le d(v,u^*_i) + d(u^*_i,f(i)) \\
		& \le 2\cdot (b(G)-1)
	\end{split}
	\label{eq:5}
\end{equation}
	
	Additionally, since the length of $f$ is $3\cdot b(G)-2$, the first $b(G)$ vertices have a covering radius of $3\cdot b(G)-2-i$, which is greater than or equal to $2\cdot (b(G)-1)$.
	
	\begin{equation}
		\forall i \in \{1,2,...,b(G)\} \ , \
		3\cdot b(G)-2-i \ge 2\cdot (b(G)-1)
		\label{eq:6}
	\end{equation}
	
	To this point, we have shown that the first $b(G)$ vertices $f(i)$ cover all the vertices that are covered by some vertex $u^*_i \in S^*$. To complete the proof, we have to show that there is a unique vertex $u^*_i$ for each vertex $f(i)$. By Observation~\ref{obs:222}, vertex $f(i)$ can cover some vertex $u^*_i \in S^*$ at a distance at most $b(G)-1$. Let us assume for a moment that vertex $u^*_i$ equals some vertex $u^*_j$, where $j<i$. In this case, by the triangle inequality, the distance from vertex $f(i)$ to vertex $f(j)$ should be at most $2 \cdot (b(G)-1)$ (Eq.~\ref{eq:s}).
	
	\begin{equation}
		d(f(i),f(j)) \le d(f(i),u^*_i) + d(u^*_j,f(j)) \le 2\cdot (b(G)-1)
		\label{eq:s}
	\end{equation}
	
	Notice that, under this scenario, all vertices in $V$ are already covered by vertex $f(j)$. The reason is that $f(i)$ is the farthest vertex from set $\{f(1),f(2),...,f(i-1)\}$, and $f(j)$ has a covering radius of at least $2\cdot (b(G)-1)$, where $j<i$. In summary, under the scenario where $u^*_i=u^*_j$, all the vertices are covered, which means that $f$ is a burning sequence of length $3 \cdot b(G)-2$. Let us now explore the case where $u^*_i$ is different from every other vertex $u^*_j$, where $j<i<b(G)$. This case is straightforward because each of the first $b(G)$ vertices in $f$ can be associated to a unique vertex $u^*_i \in S^*$.
	
	Finally, $f$ is a valid burning sequence because it maps each of the $3 \cdot b(G)-2$ positions to one vertex, all vertices in the sequence have an assigned position, and it burns all the vertices in $V$. The approximation factor of this algorithm is $3-2/b(G)$ (Eq.~\ref{eq:7}).
	
	\begin{equation}
		\frac{3\cdot b(G)-2}{b(G)}=3-\frac{2}{b(G)}
		\label{eq:7}
	\end{equation}
	
\end{proof}

\begin{theorem}
	BFF (Algorithm~\ref{alg:2}) is a $3-2/b(G)$-approximation algorithm for the graph burning problem.
	\label{theo:2}
\end{theorem}

\begin{proof}
	By Theorem~\ref{theo:1}, a sequence $f$ of length $3\cdot b(G)-2$ such that the first $b(G)$ vertices $f(i)$ are the farthest from the set $\{f(1),f(2),...,f(i-1)\}$ is a burning sequence. Namely, it covers the whole set of vertices $V$. Since the BFF algorithm constructs sequence $f:\{x \in \mathbb{Z}^+ \ | \ x \le |S| \} \rightarrow S$ by adding the farthest vertex $f(i)$ from the previously added vertices $\{f(1),f(2),...,f(i-1)\}$ (lines 9-11), at some point, it may reach a maximum length of $3\cdot b(G)-2$ that covers the whole set of vertices $V$. Since the BFF algorithm's stop condition is to cover all the vertices (line 8), this algorithm will return a valid burning sequence of length at most \mbox{$3\cdot b(G)-2$}. By Eq.~\ref{eq:7}, the BFF algorithm is a $3-2/b(G)$-approximation algorithm.
\end{proof}

\subsection{Running time}

Algorithms~\ref{alg:1} and \ref{alg:2} require computing the distance between every pair of vertices. Thus, a polynomial-time all-pairs shortest path algorithm has to be executed at the beginning (line 1). This algorithm may be based on Breadth-First Search. This way, line 1 can be executed in $O(mn + n^2)$ steps, where $n=|V|$ and $m=|E|$. For dense graphs, this complexity is $O(n^3)$. Afterwards, Algorithm~\ref{alg:1} computes the farthest vertex up to $O(b(G))$ times. The farthest vertex can be computed in $O(n)$ steps by keeping track of the last computed distances. Thus, the overall complexity of Algorithm~\ref{alg:1} is $O(n^3) + O(b(G) \cdot n)$, which is $O(n^3)$.

Algorithm~\ref{alg:2} computes the farthest vertex up to $O(b(G))$ times too. However, unlike Algorithm~\ref{alg:1}, it requires checking if all vertices are already burned (lines 12 to 14). So, at each iteration, it explores the neighborhood of the vertices that have been burned at the current iteration (line 12). The BFF algorithm checks each vertex's neighborhood only once during the whole \textit{while} cycle by keeping track of the already explored vertices. This way, checking if all vertices are burning takes $O(n^2)$ steps for the whole cycle (lines 8 to 16), and the overall complexity of the BFF algorithm is $O(n^3) + O(b(G) \cdot n) + O(n^2)$, which is $O(n^3)$ too.

The following subsection introduces a family of graphs that shows that the approximation factor of BFF (Algorithm~\ref{alg:2}) is tight. Besides, these graphs show that it may be convenient to repeat the BFF algorithm $n$ times to get better solutions. By doing this, the overall complexity of this algorithm becomes $O(n^3) + O(b(G) \cdot n^2) + O(n^3)$, which remains being $O(n^3)$. We refer to this algorithm as the BFF+ algorithm.

\subsection{Tight examples}

Theorem~\ref{theo:2} shows that BFF has an approximation factor of $3-2/b(G)$. However, this is an upper bound. A tight example is needed to determine that a better upper bound cannot exist for this algorithm. Although one tight example suffices, we present two families of them. Figures~\ref{fig:3} and \ref{fig:5} show a family of disconnected ($H_k$) and connected ($J_k$) tight examples, respectively. For every $k \in \mathbb{Z}^+$, graphs $H_k$ and $J_k$ can be constructed. Each graph $H_k$ consists of $k-1$ disconnected singletons, and $2k-1$ paths of length $k-1$ joined at vertex $v_k$. Observe that, for any value of $k \in \mathbb{Z}^+ \setminus \{1\}$, $H_k$ is a disconnected graph with $k$ components. Therefore, its burning number $b(H_k)$ cannot be smaller than $k$, i.e., at least one vertex from each connected component must be in the optimal burning sequence. It is easy to observe that $(v_k,v_{k-1},...,v_2,v_1)$ is its optimal burning sequence. Thus, $b(H_k)=k$. If BFF is executed over $H_k$, it may generate the following sequence: $(v_1,v_2,...,v_{k-1},u_1,u_2,...,u_{2k-1})$, which has length $3k-2$. Naturally, all the singletons, $v_1$, $v_2$, up to $v_{k-1}$, must be in the sequence. Then, if $u_1$ is selected, all vertices $u_2$, $u_3$, up to $u_{2k-1}$ must be in the sequence too; otherwise, vertex $u_{2k-1}$ would not be burned. In more detail, for every $i \in \{1,2,...,2k-1\}$, $u_i$'s covering radius, $r(u_i)=(3k-2)-(k-1+i)=2k-1-i$, is at most equal to its distance to $u_{2k-1}$, $d(u_i,u_{2k-1})=2(k-1)=2k-2$ (See how Eq.~\ref{eq:a1} reduces to Eq.~\ref{eq:a3}.) Therefore, if vertex $u_{2k-1}$ were removed from the sequence, it would not be burned by any other vertex; thus, such sequence would not be a valid burning sequence. Since the sequence returned by BFF has length $3k-2$ and $k=b(H_k)$, the length of this sequence is $3\cdot b(H_k)-2$. Thus, the approximation factor of BFF is tight.

\begin{eqnarray}
	r(u_i) & \le & d(u_i, u_{2k-1}) \label{eq:a1} \\
	2k-1-i & \le & 2k-2 \label{eq:a2}\\
	0 & \le & i-1 \label{eq:a3}
\end{eqnarray}

\begin{figure}[h]%
	\centering
	\includegraphics[scale=0.42]{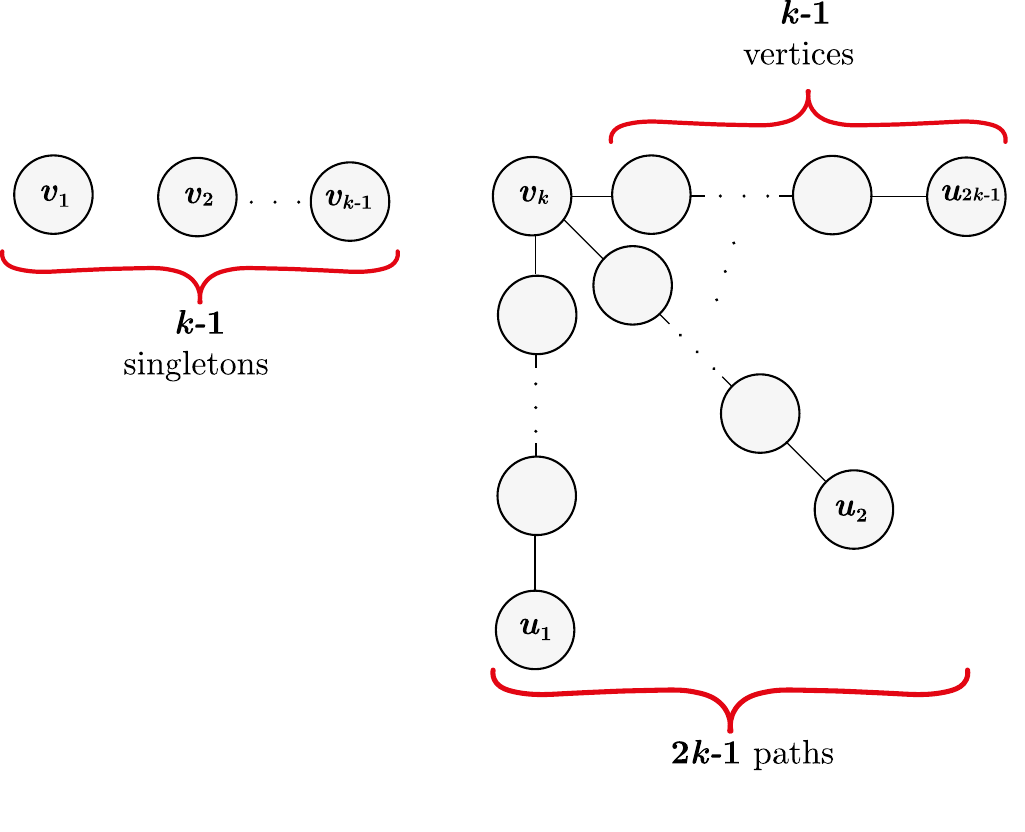}%
	\caption{A family of disconnected tight examples $H_k$ for the BFF algorithm.}
	\label{fig:3}
\end{figure}

\begin{figure}[h]%
	\centering
	\includegraphics[scale=0.42]{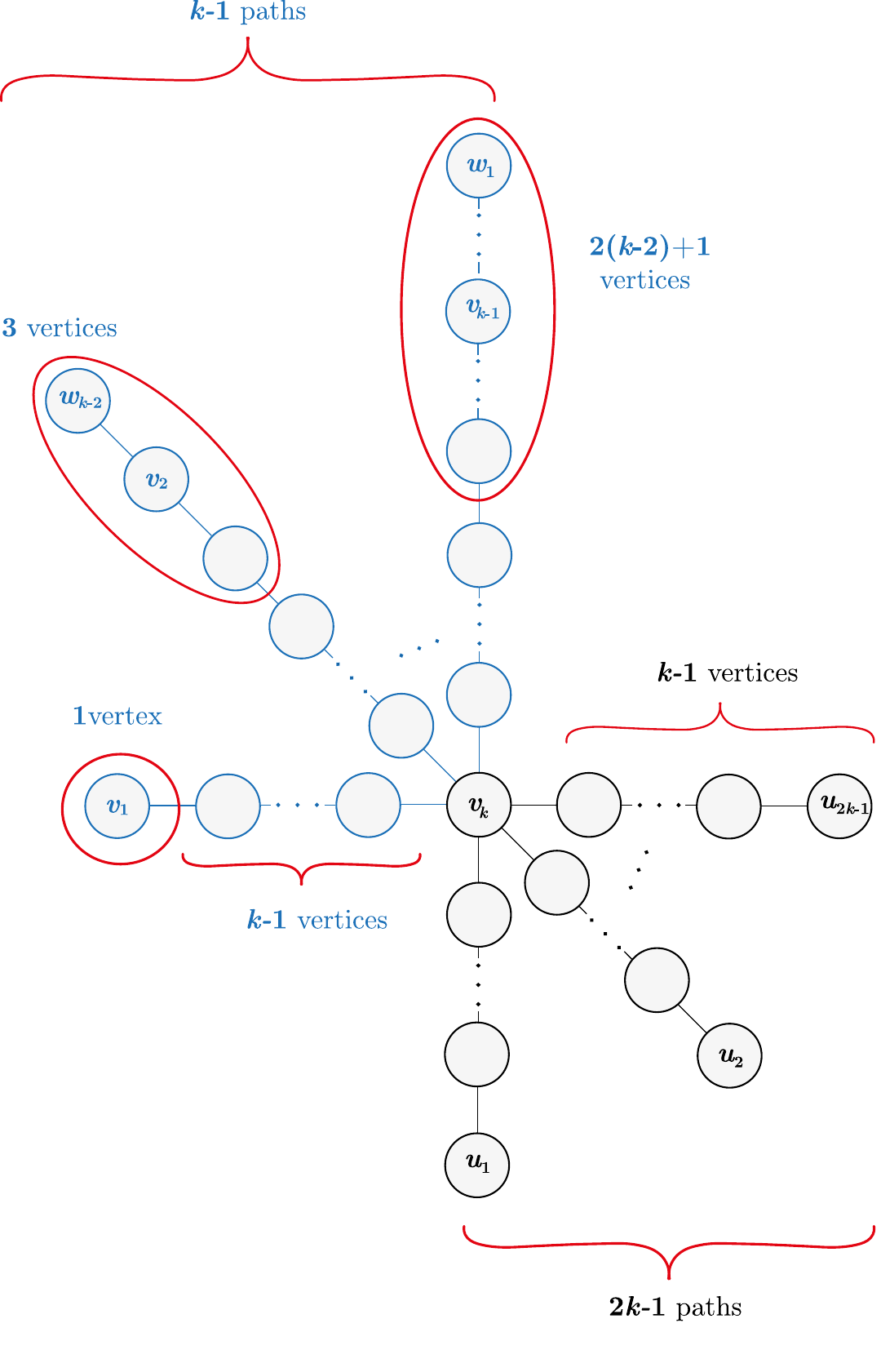}%
	\caption{A family of connected tight examples $J_k$ for the BFF algorithm.}
	\label{fig:5}
\end{figure}

For further clarification, Figure~\ref{fig:4} and Table~\ref{tab:1} show the distance from every vertex to each partial burning sequence generated by the BFF algorithm when executed over the tight example $H_3$. The optimal burning sequence of this graph is $(v_3,v_2,v_1)$. Notice that if two or more vertices are the farthest ones, any of them can be selected. In Table~\ref{tab:1}, the selected farthest vertices are in bold, and the burned vertices are marked in red. Although the distance between two vertices in different connected components is undefined, we considered such distances as maximum. Nevertheless, as mentioned in Observation~\ref{obs:2}, those undefined distances may also be replaced by a value of $n$. As expected, the length of the sequence returned by BFF is $3\cdot b(H_3)-2=3\cdot 3-2 = 7$.

%
\begin{table*}[h]
	\centering
	\caption{State of the BFF algorithm when executed over the tight example $H_3$ (See Figure~\ref{fig:4}.) The selected farthest vertices are in bold and the burned vertices are marked in red.}
	\label{tab:1}       
	\begin{tabular}{lccccccccccccc}
		\hline\noalign{\smallskip}
		& \multicolumn{13}{c}{$d(v,S)$}\\
		\cline{2-14}
		$f:\{x \in \mathbb{Z}^+  |  x \le |S|\} \rightarrow S$  & $v_1$ &  $v_2$ &  $v_3$ &  $v_4$ &  $v_5$ &  $v_6$ &  $v_7$ &  $v_8$ &  $v_9$ &   $v_{10}$ &   $v_{11}$ &  $v_{12}$ &   $v_{13}$\\
		\noalign{\smallskip}\hline\noalign{\smallskip}
		$( \ )$ & $\pmb{\infty}$ & $\infty$ & $\infty$ & $\infty$ & $\infty$ & $\infty$ & $\infty$ & $\infty$ & $\infty$ & $\infty$ & $\infty$ & $\infty$ & $\infty$\\
		
		$(v_1)$ & \textcolor{red}{$0$} & $\pmb{\infty}$ & $\infty$ & $\infty$ & $\infty$ & $\infty$ & $\infty$ & $\infty$ & $\infty$ & $\infty$ & $\infty$ & $\infty$ & $\infty$\\
		
		$(v_1,v_2)$ & \textcolor{red}{$0$} & \textcolor{red}{$0$} & $\infty$ & $\infty$ & $\infty$ & $\infty$ & $\pmb{\infty}$ & $\infty$ & $\infty$ & $\infty$ & $\infty$ & $\infty$ & $\infty$\\
		
		$(v_1,v_2,v_7)$ & \textcolor{red}{$0$} & \textcolor{red}{$0$} & $2$ & $3$ & $\pmb{4}$ & $1$ & \textcolor{red}{$0$} & $3$ & $4$ & $3$ & $4$ & $3$ & $4$\\
		
		$(v_1,v_2,v_7,v_5)$ & \textcolor{red}{$0$} & \textcolor{red}{$0$} & $2$ & $1$ & \textcolor{red}{$0$} & \textcolor{red}{$1$} & \textcolor{red}{$0$} & $3$ & $\pmb{4}$ & $3$ & $4$ & $3$ & $4$\\
		
		$(v_1,v_2,v_7,v_5,v_9)$ & \textcolor{red}{$0$} & \textcolor{red}{$0$} & \textcolor{red}{$2$} & \textcolor{red}{$1$} & \textcolor{red}{$0$} & \textcolor{red}{$1$} & \textcolor{red}{$0$} & $1$ & \textcolor{red}{$0$} & $3$ & $4$ & $3$ & $\pmb{4}$\\
		
		$(v_1,v_2,v_7,v_5,v_9,v_{13})$ & \textcolor{red}{$0$} & \textcolor{red}{$0$} & \textcolor{red}{$2$} & \textcolor{red}{$1$} & \textcolor{red}{$0$} & \textcolor{red}{$1$} & \textcolor{red}{$0$} & \textcolor{red}{$1$} & \textcolor{red}{$0$} & \textcolor{red}{$3$} & $\pmb{4}$ & \textcolor{red}{$1$} & \textcolor{red}{$0$}\\
		
		$(v_1,v_2,v_7,v_5,v_9,v_{13},v_{11})$ & \textcolor{red}{$0$} & \textcolor{red}{$0$} & \textcolor{red}{$2$} & \textcolor{red}{$1$} & \textcolor{red}{$0$} & \textcolor{red}{$1$} & \textcolor{red}{$0$} & \textcolor{red}{$1$} & \textcolor{red}{$0$} & \textcolor{red}{$1$} & \textcolor{red}{$0$} & \textcolor{red}{$1$} & \textcolor{red}{$0$}\\		
		\noalign{\smallskip}\hline
	\end{tabular}
\end{table*}

\begin{figure}[h]%
	\centering
	\includegraphics[scale=0.47]{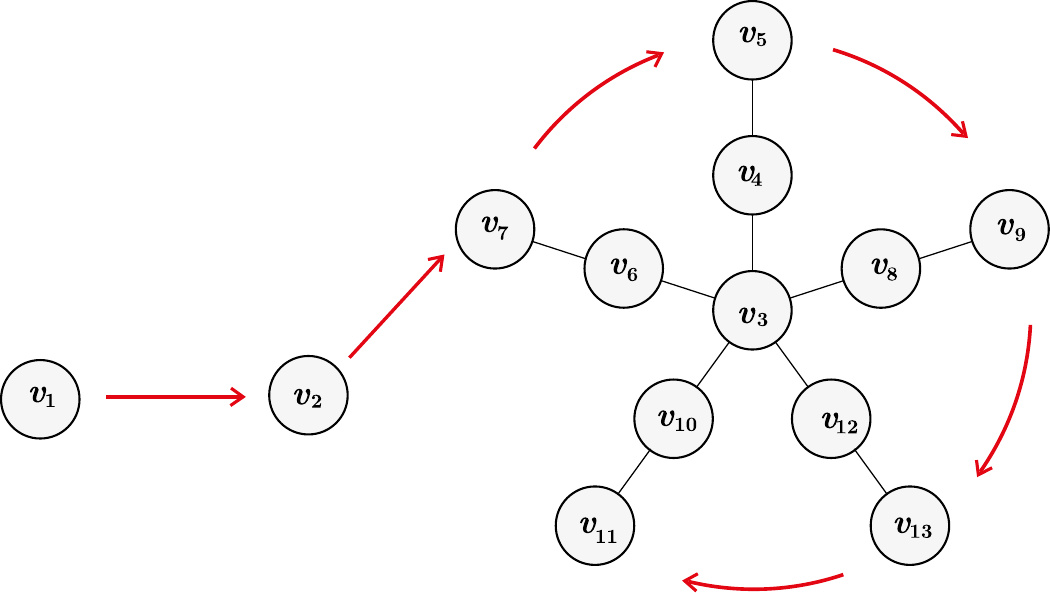}%
	\caption{Burning sequence returned by BFF over $H_3$.}
	\label{fig:4}
\end{figure}

Since a connected graph may be more interesting, Figure~\ref{fig:5} shows a family of connected tight examples $J_k$. Each graph $J_k$ consists of $2k-1$ paths of length $k-1$ joined at $v_k$ (black vertices) and $k-1$ extra paths joined at $v_k$ too (blue vertices); the length of each extra path is $(k-1)+j$, where $j$ is a positive odd number going from 1 to $2(k-2)+1$. It is not difficult to see that the optimal burning sequence of $J_k$ is $(v_k,v_{k-1},...,v_2,v_1)$. So, $b(J_k)=k$. If BFF is executed over $J_k$, it may generate the following sequence: $(w_1,w_2,...,w_{k-2},v_1,u_1,u_2,...,u_{2k-1})$, which has length $3k-2$ (Notice that $v_1$ may also be labeled as $w_{k-1}$.) All these vertices must be in the sequence; otherwise, vertex $u_{2k-1}$ would not be burned. In more detail, for every $i \in \{1,2,...,k-1\}$, $w_i$'s covering radius, $r(w_i)=3k-2-i$, is at most equal to its distance to $u_{2k-1}$, $d(w_i,u_{2k-1})=2(k-(i+1))+2(k-1)+1$ (See how Eq.~\ref{eq:b1} reduces to Eq.~\ref{eq:b3}; notice that $i<k$.) Using an argument similar to the one used for $H_k$, the covering radius $r(u_i)$ is at most $d(u_i,u_{2k-1})$ too. Therefore, if vertex $u_{2k-1}$ were removed from the sequence, it would not be burned by any other vertex; thus, such sequence would not be a valid burning sequence. Since the sequence returned by BFF has length $3k-2$ and $k=b(J_k)$, the length of this sequence is $3\cdot b(J_k)-2$.

\begin{eqnarray}
	r(w_i) & \le & d(w_i, u_{2k-1}) \label{eq:b1} \\
	3k-2-i & \le & 2(k-(i+1))+2(k-1)+1 \label{eq:b2}\\
	0 & \le & k-i-1 \label{eq:b3}
\end{eqnarray}

Figure~\ref{fig:6} and Table~\ref{tab:4} show the distance from every vertex to each partial burning sequence generated by the BFF algorithm when executed over the tight example $J_3$. The optimal burning sequence of this graph is $(v_6,v_2,v_7)$. In Table~\ref{tab:4}, the selected farthest vertices are in bold, and the burned vertices are marked in red. As expected, the length of the sequence returned by BFF is $3\cdot b(J_3)-2=3\cdot 3-2 = 7$.

%
\begin{table*}[h] 
	\centering
	\caption{State of the BFF algorithm when executed over the tight example $J_3$ (See Figure~\ref{fig:6}.) The selected farthest vertices are in bold and the burned vertices are marked in red.}
	\label{tab:4}
	\begin{tabular}{lp{0.22cm}p{0.22cm}p{0.22cm}p{0.22cm}p{0.22cm}p{0.22cm}p{0.22cm}p{0.22cm}p{0.22cm}p{0.22cm}p{0.22cm}p{0.22cm}p{0.22cm}p{0.22cm}p{0.22cm}p{0.22cm}p{0.22cm}p{0.22cm}p{0.22cm}}
		\hline\noalign{\smallskip}
		& \multicolumn{19}{c}{$d(v,S)$}\\
		\cline{2-20}
		$f:\{x \in \mathbb{Z}^+  |  x \le |S|\} \rightarrow S$  & $v_1$ &  $v_2$ &  $v_3$ &  $v_4$ &  $v_5$ &  $v_6$ &  $v_7$ &  $v_8$ &  $v_9$ &   $v_{10}$ &   $v_{11}$ &  $v_{12}$ &   $v_{13}$ & $v_{14}$ & $v_{15}$ & $v_{16}$ & $v_{17}$ & $v_{18}$ & $v_{19}$\\
		\noalign{\smallskip}\hline\noalign{\smallskip}
		
		$( \ )$ & $\pmb{\infty}$ & $\infty$ & $\infty$ & $\infty$ & $\infty$ & $\infty$ & $\infty$ & $\infty$ & $\infty$ & $\infty$ & $\infty$ & $\infty$ & $\infty$ & $\infty$ & $\infty$ & $\infty$ & $\infty$ & $\infty$ & $\infty$\\	
		
		$(v_1)$ & \textcolor{red}{$0$} & $1$ & $2$ & $3$ & $4$ & $5$ & $\pmb{8}$ & $7$ & $6$ & $6$ & $7$ & $6$ & $7$ & $6$ & $7$ & $6$ & $7$ & $6$ & $7$\\
		
		$(v_1,v_7)$ & \textcolor{red}{$0$} & \textcolor{red}{$1$} & $2$ & $3$ & $4$ & $3$ & \textcolor{red}{$0$} & $1$ & $2$ & $4$ & $\pmb{5}$ & $4$ & $5$ & $4$ & $5$ & $4$ & $5$ & $4$ & $5$\\
		
		$(v_1,v_7,v_{11})$ & \textcolor{red}{$0$} & \textcolor{red}{$1$} & \textcolor{red}{$2$} & $3$ & $3$ & $2$ & \textcolor{red}{$0$} & \textcolor{red}{$1$} & $2$ & $1$ & \textcolor{red}{$0$} & $3$ & $\pmb{4}$ & $3$ & $4$ & $3$ & $4$ & $3$ & $4$\\
		
		$(v_1,v_7,v_{11},v_{13})$ & \textcolor{red}{$0$} & \textcolor{red}{$1$} & \textcolor{red}{$2$} & \textcolor{red}{$3$} & $3$ & $2$ & \textcolor{red}{$0$} & \textcolor{red}{$1$} & \textcolor{red}{$2$} & \textcolor{red}{$1$} & \textcolor{red}{$0$} & $1$ & \textcolor{red}{$0$} & $3$ & $\pmb{4}$ & $3$ & $4$ & $3$ & $4$\\
		
		$(v_1,v_7,v_{11},v_{13},v_{15})$ & \textcolor{red}{$0$} & \textcolor{red}{$1$} & \textcolor{red}{$2$} & \textcolor{red}{$3$} & \textcolor{red}{$3$} & \textcolor{red}{$2$} & \textcolor{red}{$0$} & \textcolor{red}{$1$} & \textcolor{red}{$2$} & \textcolor{red}{$1$} & \textcolor{red}{$0$} & \textcolor{red}{$1$} & \textcolor{red}{$0$} & $1$ & \textcolor{red}{$0$} & $3$ & $\pmb{4}$ & $3$ & $4$\\
		
		$(v_1,v_7,v_{11},v_{13},v_{15},v_{17})$ & \textcolor{red}{$0$} & \textcolor{red}{$1$} & \textcolor{red}{$2$} & \textcolor{red}{$3$} & \textcolor{red}{$3$} & \textcolor{red}{$2$} & \textcolor{red}{$0$} & \textcolor{red}{$1$} & \textcolor{red}{$2$} & \textcolor{red}{$1$} & \textcolor{red}{$0$} & \textcolor{red}{$1$} & \textcolor{red}{$0$} & \textcolor{red}{$1$} & \textcolor{red}{$0$} & \textcolor{red}{$1$} & \textcolor{red}{$0$} & \textcolor{red}{$3$} & $\pmb{4}$\\
		
		$(v_1,v_7,v_{11},v_{13},v_{15},v_{17},v_{19})$ & \textcolor{red}{$0$} & \textcolor{red}{$1$} & \textcolor{red}{$2$} & \textcolor{red}{$3$} & \textcolor{red}{$3$} & \textcolor{red}{$2$} & \textcolor{red}{$0$} & \textcolor{red}{$1$} & \textcolor{red}{$2$} & \textcolor{red}{$1$} & \textcolor{red}{$0$} & \textcolor{red}{$1$} & \textcolor{red}{$0$} & \textcolor{red}{$1$} & \textcolor{red}{$0$} & \textcolor{red}{$1$} & \textcolor{red}{$0$} & \textcolor{red}{$1$} & \textcolor{red}{$0$}\\		
		
		\noalign{\smallskip}\hline
	\end{tabular}
\end{table*}

\begin{figure}[h]
	\centering
	\includegraphics[scale=0.47]{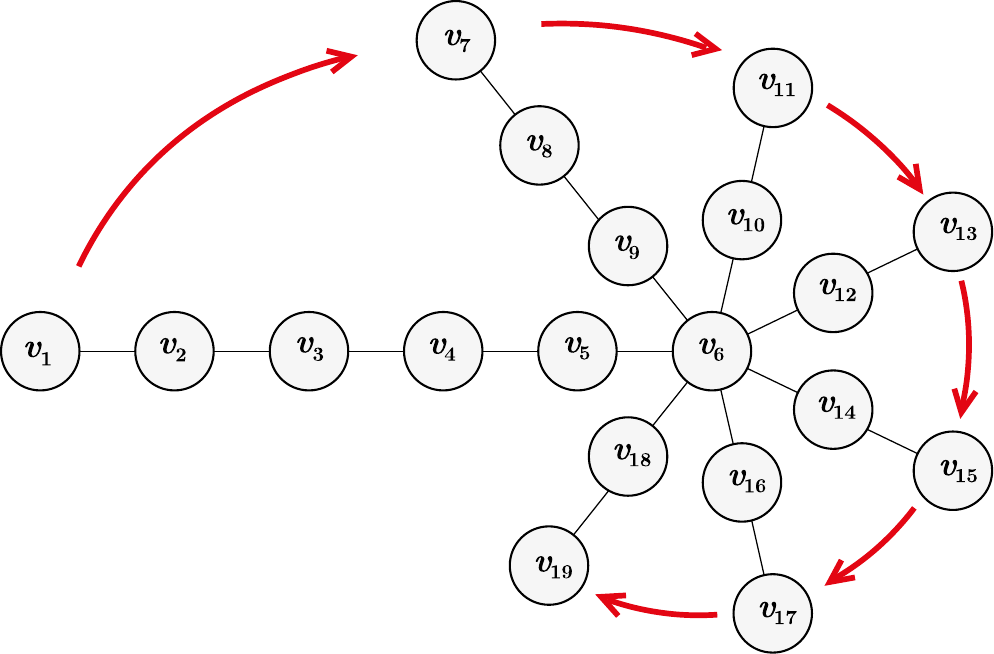}%
	\caption{Burning sequence returned by BFF over $J_3$.}
	\label{fig:6}
\end{figure}


The families of tight examples $H_k$ and $J_k$ show how selecting a bad initial vertex may lead to a bad burning sequence. Therefore, it could be helpful to repeat the BFF algorithm with a different initial vertex every time. We refer to this algorithm as BFF+. As mentioned in the previous subsection, BFF+ has complexity $O(n^3)$.

Finally, notice that $H_k$ and $J_k$ with $k=1$ are the singleton graph $K_1=(V,E)$, where $V=\{v_1\}$ and $E=\emptyset$. Naturally, $b(K_1) = 1$. Since BFF returns a burning sequence of length at most $3\cdot b(K_1)-2=3\cdot 1-2=1$, it has no other choice than to return the optimal burning sequence. The tight example $K_1$ lets us observe that the best possible approximation algorithm for the graph burning problem that returns solutions of form $3\cdot b(G)-\beta$ cannot have $\beta > 2$. Otherwise, that would imply that the optimal burning sequence of $K_1$ is not optimal, i.e., a contradiction.

\subsection{A brief empirical evaluation}
\label{sec:emp}

The goal of BFF never was to compete against much more elaborated heuristics from the literature. Since BFF can almost triple the length of an optimal solution, it seems unlikely that it may be practical at all. However, after some experimentation, we found that it tends to find near-optimal solutions. In this section, we present the results of such empirical investigation.

Tables~\ref{tab:2} and \ref{tab:3} show a brief comparison among BFF, BFF+, the 3-approximation algorithm of Bonato and Kamali (Bon) \cite{bonato2019approximation}, the Greedy Algorithm with Forward-Looking Search Strategy (GFSS) \cite{vsimon2019heuristics}, and the Component-Based Recursive heuristic (CBR) \cite{gautam2021faster}. These two last heuristics are among the most efficient to date. The length of the burning sequences found by the Bon, GFSS, and CBR algorithms was taken from the literature \cite{gautam2021faster}.

Tables~\ref{tab:2} and \ref{tab:3} show the length of the burning sequences returned by the tested algorithms. The used instances are real-world graphs from the \textit{network repository} \cite{nr} and the Stanford large network dataset collection (SNAP) \cite{leskovec2016snap}. The BFF and BFF+ algorithms were implemented in C++. These implementations were executed on an Asus laptop with an Intel Core i5-8300H processor of 2.3 GHz, 8 GB of RAM memory, and Ubuntu 21.04 operating system. They can be consulted from \url{https://github.com/alex-cornejo/bff_alg}. 

\begin{table}[h]
	\centering
	\caption{Burning sequence length returned by the tested algorithms over some graphs from the \textit{network repository}. Best solutions are highlighted.}
	\label{tab:2}       
	\resizebox{11cm}{!}{
	\begin{tabular}{ccc|ccccc}
		\hline
		&&& \multicolumn{5}{c}{burning sequence length} \\
		\cline{4-8}
		$G=(V,E)$  &  $|V|$  &  $|E|$  &  Bon  & GFSS & CBR & BFF & BFF+ \\
		\hline
		ca-netscience & 379 & 914 & 12 & \textcolor{red}{\textbf{7}} & \textcolor{red}{\textbf{7}} & 8 & 8 \\
		web-polblogs & 643 & 2,280 & 9 & \textcolor{red}{\textbf{6}} & \textcolor{red}{\textbf{6}} & 8 & \textcolor{red}{\textbf{6}} \\		
		socfb-Reed98 & 962 & 18,812 & 6 & \textcolor{red}{\textbf{4}} & \textcolor{red}{\textbf{4}} & 5 & \textcolor{red}{\textbf{4}} \\			
		econ-mahindas & 1,258 & 7,513 & 9 & \textcolor{red}{\textbf{5}} & \textcolor{red}{\textbf{5}} & 6 & \textcolor{red}{\textbf{5}} \\			
		cite-DBLP & 12,591 & 49,743 & 120 & \textcolor{red}{\textbf{41}} & \textcolor{red}{\textbf{41}} & \textcolor{red}{\textbf{41}} & \textcolor{red}{\textbf{41}} \\	
		\hline
		\noalign{\smallskip}
	\end{tabular}
}
\end{table}

\begin{table}[h]
	\centering
	\caption{Burning sequence length returned by the tested algorithms over some graphs from the SNAP dataset. Best solutions are highlighted.}
	\label{tab:3}       
	\resizebox{11cm}{!}{
	\begin{tabular}{ccc|ccccc}
		\hline
		&&& \multicolumn{5}{c}{burning sequence length} \\
		\cline{4-8}
		$G=(V,E)$  &  $|V|$  &  $|E|$  &  Bon  & GFSS & CBR & BFF & BFF+ \\
		\hline
		chameleon & 2,277 & 31,421 & 9 & \textcolor{red}{\textbf{6}} & \textcolor{red}{\textbf{6}} & 8 & \textcolor{red}{\textbf{6}} \\	
		TVshow & 3,892 & 17,262 & 18 & \textcolor{red}{\textbf{10}} & \textcolor{red}{\textbf{10}} & 11 & \textcolor{red}{\textbf{10}}\\
		ego-facebook & 4,039 & 88,234 & 9 & \textcolor{red}{\textbf{4}} & \textcolor{red}{\textbf{4}} & 6 & 5\\
		squirrel & 5,201 & 198,493 & 9 & \textcolor{red}{\textbf{6}} & \textcolor{red}{\textbf{6}} & 7 & \textcolor{red}{\textbf{6}} \\
		politician & 5,908 & 41,729 & 12 & \textcolor{red}{\textbf{7}} & \textcolor{red}{\textbf{7}} & 9 & \textcolor{red}{\textbf{7}}\\
		government & 7,057 & 89,455 & 9 & \textcolor{red}{\textbf{6}} & \textcolor{red}{\textbf{6}} & 7 & \textcolor{red}{\textbf{6}} \\
		crocodile & 11,631 & 170,918 & 12 & \textcolor{red}{\textbf{6}} & \textcolor{red}{\textbf{6}} & 8 & \textcolor{red}{\textbf{6}}\\	
		\hline
		\noalign{\smallskip}
	\end{tabular}
}
\end{table}

Tables~\ref{tab:2} and \ref{tab:3} show that BFF and BFF+ returned better solutions than Bon in all cases. Besides, in ten out of twelve instances, BFF+ found solutions of the same length as those found by the GFSS and CBR heuristics. Finally, the $O(n^3)$ complexity of BFF and BFF+ might be considered impractical for large instances. Nevertheless, the execution time of our BFF (BFF+) implementation is \textasciitilde10ms (\textasciitilde20ms) over the smallest instance ca-netscience (379 vertices) and \textasciitilde20s (\textasciitilde38s) over the biggest instance cite-DBLP (12,591 vertices). These execution times consider the time consumed by computing the all-pairs shortest path.

\section{Conclusions}
\label{sec:con}
The graph burning problem is an NP-hard problem that models different phenomena, such as social contagion, the sequential spread of information, or viral infections under a very idealistic context \cite{bonato2014burning,vsimon2019heuristics,gupta2021burning}. This problem's primary attribute is that it helps quantify how vulnerable a graph is to contagion. This problem receives as input a simple graph $G=(V,E)$ and seeks a burning sequence of minimum length. The length of this sequence is known as the burning number $b(G)$. Intuitively, a shorter $b(G)$ implies that $G$ is more vulnerable to contagion. Due to its nature, the graph burning problem has been approached through approximation algorithms and heuristics \cite{bonato2019approximation,vsimon2019heuristics,gautam2021faster}.

This paper introduces two simple approximation algorithms for the graph burning problem (Algorithms~\ref{alg:1} and \ref{alg:2}). The approximation factor of both is $3 -2/b(G)$. Namely, given a simple graph $G=(V,E)$, they generate burning sequences of length at most $3 \cdot b(G) -2$. Actually, the goal of Algorithm~\ref{alg:1} is to smooth the way for understanding Algorithm~\ref{alg:2}, which represents the main contribution of this paper. We refer to Algorithm~\ref{alg:2} as BFF. Both algorithms require an initialization that consists of finding the shortest path between every pair of vertices. These distances may be computed by some polynomial-time shortest path algorithm, such as a Breadth-First Search-based algorithm \cite{kleinberg2006algorithm}. Considering this initialization, the overall complexity of Algorithms~\ref{alg:1} and \ref{alg:2} is $O(n^3)$. Besides, the BFF algorithm can be repeated $n$ times to generate better solutions (BFF+ algorithm). The overall complexity of BFF+ is $O(n^3)$ too. 

The BFF and BFF+ algorithms were implemented and executed over a small set of instances from the literature. The obtained results show that these algorithms generate near-optimal burning sequences. Finally, future work includes exploring better approximation algorithms and how the graph burning problem relates to other NP-hard center selection problems.


\end{document}